\documentclass[10pt,a4paper]{article}
\usepackage{amssymb,amsmath}
\usepackage{bbm}
\usepackage{stmaryrd}
\usepackage[mathcal]{euscript}
\usepackage{hyperref, diagrams}

\numberwithin{equation}{section}

\newcounter{thMM}
\newcounter{leMM}
\newcounter{deFF}
\newcounter{exMP}
\newcounter{prOP}
\newcounter{coRR}
\newcounter{coNS}
\newenvironment{theorem}[1][Theorem]{\refstepcounter{thMM}\trivlist
   \item[\hskip19pt{\bf #1~\arabic{thMM}.}]\it\hskip3pt}{\endtrivlist}
\newenvironment{lemma}[1][Lemma]{\refstepcounter{leMM}\trivlist
   \item[\hskip19pt{\bf #1~\arabic{leMM}.}]\it\hskip3pt}{\endtrivlist}
\newenvironment{definition}[1][Definition]{\refstepcounter{deFF}\trivlist
   \item[\hskip19pt{\bf #1~\arabic{deFF}.}]\rm\hskip3pt}{\endtrivlist}
\newenvironment{proposition}[1][Proposition]{\refstepcounter{prOP}\trivlist
   \item[\hskip19pt{\bf #1~\arabic{prOP}.}]\it\hskip3pt}{\endtrivlist}
\newenvironment{corollary}[1][Corollary]{\refstepcounter{coRR}\trivlist
   \item[\hskip19pt{\bf #1~\arabic{coRR}.}]\it\hskip3pt}{\endtrivlist}
\newenvironment{construction}[1][Construction]{\refstepcounter{coNS}\trivlist
   \item[\hskip19pt{\bf #1~\arabic{coNS}.}]\rm\hskip3pt}{\endtrivlist}

\newenvironment{proof}[1][Proof]{\begin{trivlist}
\item[\hskip \labelsep {\bfseries #1}] }{ \vspace{-15pt}\begin{flushright} $\square$\end{flushright}\end{trivlist}}
\newenvironment{remark}[1][Remark]{\begin{trivlist}
\item[\hskip \labelsep {\bfseries #1}]\hskip3pt}{\end{trivlist}}

\newenvironment{example}[1][Example]{\begin{trivlist}
\item[\hskip \labelsep {\bfseries #1}]}{\end{trivlist}}
\newenvironment{statement}[1][Statement]{\begin{trivlist}
\item[\hskip \labelsep {\bfseries #1}]\it  }{ \end{trivlist}}

\newcommand{\Id}{\mathbbmss{1}}

\newcommand{\InHom}{\mbox{$\underline{\Hom}$}}

\newcommand{\kAbove}[2]{\overset{\left. \scriptscriptstyle #1 \hspace{4pt}\right. }{\vphantom{a}\smash{#2}}}

\newcommand{\InDiff}{\mbox{$\underline{\Diff}$}}

\DeclareMathOperator{\Vect}{Vect}

\DeclareMathOperator{\Diff}{Diff}

\DeclareMathOperator{\Hom}{Hom}
\DeclareMathOperator{\Spec}{Spec}

\newcommand{\catname}[1]{\textnormal{\texttt{#1}}}

\begin{document}
\bibliographystyle{plain}

\author{Andrew James Bruce\\ \small \emph{Institute of Mathematics, Polish Academy of Sciences}\\ \small email:\texttt{andrewjamesbruce@googlemail.com}}

\date{\today}
\title{On curves and jets of curves on supermanifolds}
\maketitle

\begin{abstract}
In this  note we examine a natural concept of a curve on a supermanifold and the subsequent notion of the jet of a curve. We then tackle the question of geometrically defining the higher order tangent bundles of a supermanifold. Finally we make a quick comparison with the notion of a curve presented here are other common notions found in the literature.
\end{abstract}

\begin{small}
\end{small}
\section{Introduction}

Supermanifolds are a generalisation of the notion of a smooth manifold in which the structure sheaf of the manifold gets replaced by a sheaf of supercommutative algebras. Informally one can think of a supermanifold as a ``manifold" with both commuting and anticommuting coordinates. Initial interest in supermanifolds came from developments in theoretical physics where anticommuting variables were found to be essential in quasi-classical descriptions of fermions and Faddeev--Popov ghosts. Today the theory of supergeometry had developed into a rich area of pure mathematics in its own right. In this note we look at a geometric or kinematic definition of the higher order tangent bundle of a supermanifold.\\

Higher order tangent bundles, not to be confused with the related iterated tangent bundles, are the natural geometric home of higher derivative Lagrangian mechanics and thus a clear geometric understanding of the super-versions is important for general supermechanics. The Ostrogradski\u{\i} instability means that higher derivative Lagrangians cannot be viewed as fundamental theories, however they can serve as \emph{effective theories}.  As a side remark, there has been some renewed interest in higher derivative supersymmetric field theories in relation to supergravity effective actions from string theory \cite{VanProeyen:2013}.  Higher order tangent bundles are also fundamental to the notion of graded bundles and homogeneity structures as studied by Grabowski and Rotkiewicz \cite{Grabowski:2011}.\\

To the authors' knowledge the k-th order tangent bundle $T^{(k)}M$ of a supermanifold  $M$ was first described by  Cari\~{n}ena \& Figueroa \cite{Carinena:1998} via a ``diaginalisation" of the k-th iterated tangent bundle. This line of reasoning parallels the classical constructions, but the relation with curves and their jets is obscured. Another approach  to the k-th order tangent bundle is to consider them as Weil bundles, the classical case can be found in the monograph of Kol\'{a}r, Michor and Slovak \cite{Kolar:1993}. Essentially the k-th order tangent bundle can be viewed  as the supermanifold of all maps (both even and odd) from $\Spec(\mathbb{D}_{k})$  to the supermanifold $M$, where $\mathbb{D}_{k} = \mathbb{R}[\delta]/ \langle \delta^{k+1} \rangle$ and $\delta$ is a single even indeterminant. Details of Weil bundles on supermanifolds can be found in the work of Alldridge \cite{Alldridge:2011} and the unpublished notes of Rotkiewicz. This approach is elegant and has been extremely powerful in classical differential geometry. Note that odd maps are required in this algebraic approach and that this already complicates the situation as compared to the classical case.  \\

In this work we construct the k-th order tangent bundle  of a supermanifold in terms of superised versions of curves and their jets. To do this we employ Grothendieck's functor of points.  Loosely, the S-points of the k-th order tangent bundle are identified with the k-th jets of curves at the S-points of the supermanifold. We show that the operational handling of jets of curves on supermanifolds in terms of Taylor expansions in local coordinates can be put on  solid footing using the functor of points. To the authors knowledge this has not been properly presented in the literature before.\\

\subsection*{Preliminary notation}
For an overview of category theory we suggest Mac Lane's book \cite{MacLane:1988}. We will follow the ``Russian school" and denote by $\catname{SM}$ the category of real finite dimensional supermanifolds understood as locally superringed spaces, see for example  \cite{Carmeli:2011,Deligne:1999,Manin1997,Varadarajan2004}.  We will simply denote a supermanifold by $M$, where we understand it to be defined by its structure sheaf $(|M|, \mathcal{O}_{M})$, where $|M|$ is the reduce manifold or body of $M$. By an open superdomain $U \subset M$, we mean an open neighborhood of some point $p \in |M|$ together with the corresponding restriction of the structure sheaf.    Sections of the structure sheaf will be called functions on $M$ and will be denoted by $C^{\infty}(M)$. A morphism between supermanifolds $\phi : M \rightarrow N$ is a pair of morphisms $(|\phi|, \phi^{*})$ where $|\phi|: |M| \rightarrow |N|$ is a continuous map and $\phi^{*}: \mathcal{O}_{N} \rightarrow \mathcal{O}_{M}$ is a sheaf morphism above $|\phi|$.  The set of morphisms between  the pair of supermanifolds will be denoted by  $\Hom_{\catname{SM}}(M,N) := \Hom(M,N)$. Note that these categorical morphisms necessarily preserve the $\mathbb{Z}_{2}$-grading.  We  assume the reader has some familiarity with  the theory of supermanifolds.\\

\noindent \textbf{The functor of points:}
We will employ the Grothendieck's functor of points applied to supergeometry throughout this work, see for example \cite{Carmeli:2011,Deligne:1999}. The S-points of a supermanifold $M$ are elements in the set $\Hom(S,M)$, where $S$ is some arbitrary supermanifold. That is, one can view a supermanifold as a functor
 \begin{eqnarray}
 \nonumber M : \catname{SM}^{\textnormal{o}} &\rightarrow& \catname{Set} \\
  \nonumber S &\mapsto& \Hom(S, M) := M(S),
 \end{eqnarray}
 which is an example of the Yoneda embedding. Via Yoneda's lemma, we can identify a supermanifold with such a functor and morphisms between supermanifolds are equivalent to natural transformations between the corresponding functors. Such natural transformations amount to maps between the respective sets of S-points.   Informally one can think about the ``points" of $M$ as being parameterised by \emph{all} supermanifold $S$.

\begin{remark}
 Following the work of   Schwarz  and Voronov \cite{Shvarts:1984,Voronov:1984} it is known that it is actually sufficient  to consider just $\Lambda$-points, that is  supermanifolds of the form $\mathbb{R}^{0|q}$ $(q \geq 1)$ as paramaterisations.
\end{remark}

One can think about the evaluation of a given function at an S-point, which is the analogue of the evaluation of a function on a manifold at a point. First note that we have the bijection between $C^{\infty}(M)$ and $\Hom(M,\mathbb{R}^{1|1})$ simply given by $f \mapsto ( t \circ f , \tau \circ f)$, where we have local coordinates $(t, \tau)$ on $\mathbb{R}^{1|1}$. Then we define the value of the function $f$ at a specified  S-point $m \in M(S)$ as  $f_{m} := f \circ m \in \Hom(S, \mathbb{R}^{1|1}) \simeq C^{\infty}(S)$.  \\

As the functor of points involves maps between finite dimensional supermanifolds one can consider S-points locally via coordinates. In particular, let us employ some coordinate system $(x^{A}) =(x^{a}, \theta^{i})$ on $U \subset M$, then the S-points are then specified by systems of functions $(x_{S}^{a}, \theta_{S}^{i})$ where
 $x_{S}^{a} \in C^{\infty}(S)_{0}$ and $\theta_{S}^{i} \in C^{\infty}(S)_{1}$. As  the supermanifold $S$ is chosen arbitrarily dependence on the local coordinates of $S$ will not explicitly be presented.  Given a morphism  $\psi \in \Hom(P,S)$  between supermanifolds $P$ and $S$ we have an induced map
 \begin{eqnarray}
 \nonumber \bar{\psi} : M(S) &\rightarrow& M(P)\\
  \nonumber m &\rightarrow& m\circ \psi,
\end{eqnarray}
where $m \in M(S)$.\\

\noindent \textbf{Generalised supermanifolds and the internal homs:}
A \emph{generalised supermanifold} is an object in the functor category $ \widehat{\catname{SM}} := \catname{Fun}(\catname{SM}^{o} , \catname{Set})$, whose objects are functors from $\catname{SM}^{o}$ to the category $\catname{Set}$ and whose morphisms are natural transformations. Note that this functor category contains $\catname{SM}$ as a full subcategory via the Yoneda embedding. One passes from the category of finite dimensional supermanifolds to the larger category of generalised supermanifolds in order to understand the internal hom objects. In particular there always exists a generalised supermanifold such that the so called adjunction formula holds

\begin{equation*}
\InHom(M,N)(\bullet) :=  \Hom(\bullet \times M,N) \in \widehat{\catname{SM}}.
\end{equation*}

By some abuse of language, the mapping supermanifold $\InHom(M,N)$ is referred to as an \emph{internal hom object}, remembering that it lives in the larger category of generalised supermanifolds.  Heuristically, one should think of enriching the morphisms between supermanifolds to now have the structure of a supermanifold, however to understand this one passes to a larger category.   In essence we will use the above to define what we mean by a mapping supermanifold and will probe it using the functor of points. A generalised supermanifold is \emph{representable} if it is naturally isomorphic to a supermanifold in the image of the Yoneda embedding. For example, it is easy to see that $\Hom(\{pt\} , M ) = |M|$ while $\InHom(\{pt\},M)=M$. Another well-know example of a representable generalised supermanifold  is the antitangent bundle $\InHom(\mathbb{R}^{0|1}, M) = \Pi TM$.\\

\section{Superfunctions on $\mathbb{R}$ and their jets}\label{sec:superfunctions}

Smooth functions $\mathbb{R} \rightarrow \mathbb{R}$ play a fundamental role in the notion of classical jets on manifolds. The smooth maps $\mathbb{R} \rightarrow \mathbb{R}^{1|1}$, where we allow both even and odd maps, \emph{superfields} in the physics language,  play the analogue role in supergeometry. \\

The mapping supermanifold in question here is $ \InHom(\mathbb{R}, \mathbb{R}^{1|1})$. As this is not a set one has to take a little care with defining its ``elements". To do this  we ``probe" the mapping supermanifold  via the functor of points. That is it will be useful to consider the set

\begin{equation}
 \InHom(\mathbb{R}, \mathbb{R}^{1|1})(S) := \Hom(S \times \mathbb{R}, \mathbb{R}^{1|1}),
\end{equation}

for an arbitrary supermanifold $S$.

\begin{definition}\label{def:superfunction}
A \emph{superfunction on} $\mathbb{R}$ \emph{paramaterised by} $S \in \catname{SM}$ is a smooth map $\Upsilon_{S} \in \InHom(\mathbb{R}, \mathbb{R}^{1|1})(S)$.
\end{definition}

\begin{remark}
By convention a  function on a supermanifold is  a section of the structure sheaf, that is a smooth parity preserving map from the supermanifold to $\mathbb{R}^{1|1}$. Thus we can identify $C^{\infty}(M) = \Hom(M, \mathbb{R}^{1|1})$. No confusion between this notion and a superfunction on the real line should occur.
\end{remark}

\begin{construction}\label{con:morphism}
Let $\psi \in \Hom(P,S)$ be a morphism of supermanifolds. Then such a morphism induces a map
\begin{eqnarray}
\nonumber \Psi : \InHom(\mathbb{R}, \mathbb{R}^{1|1})(S) &\rightarrow& \InHom(\mathbb{R}, \mathbb{R}^{1|1})(P)\\
\Upsilon_{S} &\mapsto& \Upsilon_{P} :=  \Upsilon_{S} \circ (\psi \times \Id_{\mathbb{R}} ).
\end{eqnarray}
\end{construction}

Note that a superfunction on $\mathbb{R}$ parameterised by $S$ is nothing but a function on the supermanifold $S \times \mathbb{R}$. Thus we have a well defined notion of taking the derivative with respect to ``time" by picking a global coordinate on $\mathbb{R}$. In particular we can Taylor expand any superfunction with respect to any point $t_{0} \in \mathbb{R}$.

\begin{construction}\label{con:jet}
Let $\Upsilon_{S} \in \InHom(\mathbb{R}, \mathbb{R}^{1|1})(S)$ be a superfunction on $\mathbb{R}$ parameterised by $S$. The k-th jet of $\Upsilon_{S}$ at the point $t = t_{0} \in \mathbb{R}$ is the polynomial with coefficients in $C^{\infty}(S)$
\begin{equation*}
\left(J_{t_{0}}^{k}\Upsilon_{S} \right)z = \left.\Upsilon_{S}\right|_{t_{0}} + z \left.\frac{\partial \Upsilon_{S}}{\partial t}\right|_{t_{0}} + z^{2}\frac{1}{2!}\left.\frac{\partial^{2} \Upsilon_{S}}{\partial t^{2}}\right|_{t_{0}} + \cdots + z^{k}\frac{1}{k!}\left.\frac{\partial^{k} \Upsilon_{S}}{\partial t^{k}}\right|_{t_{0}},
\end{equation*}
for any natural $k$. Here $z$ is understood as an even indeterminant.
\end{construction}

\begin{definition}\label{def:equivalence}
 Let $\Upsilon_{S}, \bar{\Upsilon}_{S} \in \InHom(\mathbb{R}, \mathbb{R}^{1|1})(S)$ be  superfunctions on $\mathbb{R}$ parameterised by $S$. Then $\Upsilon_{S}$ and $\bar{\Upsilon}_{S}$ are said to be\emph{ equivalent to order} $r$ \emph{at the point} $t_{0}\in \mathbb{R}$ if and only if
 \begin{equation*}
 \left(J_{t_{0}}^{r}\Upsilon_{S}  \right) = \left(J_{t_{0}}^{r}\bar{\Upsilon}_{S} \right).
 \end{equation*}
 Clearly this is an equivalence relation on the set $\InHom(\mathbb{R}, \mathbb{R}^{1|1})(S)$.
\end{definition}

Essentially the above definition says that two parameterised superfunctions are equivalent to order $r$ at $t_{0}$ if they can be identified at this point and so can their first $r$ derivatives with respect to $t$.

\begin{lemma}\label{lem:naturality}
Let $\Upsilon_{S} \in \InHom(\mathbb{R}, \mathbb{R}^{1|1})(S)$ be a superfunction on $\mathbb{R}$. Consider an arbitrary homomorphism of supermanifolds $\psi \in \Hom(P,S)$.  Then we have
\begin{equation*}
J_{t_{0}}^{k}\Upsilon_{P} = \left(J_{t_{0}}^{k}\Upsilon_{S}  \right) \circ \psi
\end{equation*}

\end{lemma}
\begin{proof}
We have
\begin{equation*}
\Upsilon_{P}(t):= \left(\Upsilon_{S}\circ (\psi \times \Id_{R})\right)(t) = \Upsilon_{S}(t) \circ \psi,
\end{equation*}
for all $t \in \mathbb{R}$. Then taking the derivative with respect to $t$ an arbitrary number of times yields
\begin{equation*}
\frac{\partial^{r} \Upsilon_{P}(t)}{ \partial  t^{r}\hfill } = \frac{\partial^{r}}{\partial t^{r}} \left(\Upsilon_{S}(t) \circ \psi  \right) = \left( \frac{\partial^{r} \Upsilon_{S}(t)}{\partial  t^{r} \hfill}\right)\circ \psi,
\end{equation*}
as the morphism $\psi$ is independent of $t$. Then one obtains
\begin{equation*}
\left.\frac{\partial^{r} \Upsilon_{P}}{\partial t^{r}\hfill}\right|_{t_{0}} = \left( \left.\frac{\partial^{r} \Upsilon_{S}}{\partial  t^{r} \hfill}\right|_{t_{0}}\right) \circ \psi.
\end{equation*}
Then as the morphism $\psi$  is also independent of $z$ we obtain the desired result.
\end{proof}

\section{Curves on supermanifolds and higher order tangent bundles}\label{sec:curveshighertangent}
We are now in a position to describe curves on supermanifolds and the notion of their jets.  We will attempt to follow classical reasoning, say c.f. \cite{Kolar:1993}, as much as possible. There is a  clear problem with the na\"{i}ve definition of a curve as a morphism of supermanifolds $\mathbb{R} \rightarrow M$ and the subsequent construction of their jets. Such a  classical definition of a curve, as  morphisms of supermanifolds preserve Grassmann parity, totally misses the odd dimensions of the supermanifold. The only points that such classical curves can pass through are the topological points of the underlaying reduced manifold. Thus if one attempts to follow the classical constructions more or less identically, then the resulting structures are classical ones on the underlying reduced manifold. The resolution of these issues is to employ the internal Homs and the functor of points in our constructions.\\

Informally a curve on a supermanifold $M$  is a smooth map $\gamma \in \InHom(\mathbb{R}, M)$. To make proper sense of this  we  ``probe"  $\InHom(\mathbb{R}, M)$ using the functor of points.

\begin{definition}\label{def:param curve supermanifold}
A \emph{curve on a supermanifold} $M$ \emph{parameterised by another supermanifold} $S$ is a smooth morphism $\gamma_{S} \in \InHom(\mathbb{R}, M)(S)$. We will refer to $\mathbb{R}$ as the source, $M$ as the target and $S$ as the parameterisation of the curve. For brevity we will refer to  $\gamma_{S}$ as an \emph{S-curve}.
\end{definition}

In practice S-curves may only be locally defined on $M$, but we will not make an issue of this here. Note that we do not really ever deal with \emph{a single} map, but rather always a family. Note also that as $\InHom(\mathbb{R}, M)(S) = \Hom(S \times \mathbb{R} ,M)$ the image set of an S-curve $\gamma_{S}(\mathbb{R}) \subset M(S)$ is a collection of S-points of $M$. This is close to classical thinking as an S-curve traces out S-points. A little more carefully, we can view $\gamma|_{t_{0}}: \InHom(\mathbb{R}, M) \rightarrow M$ as a natural transformation via restriction of an S-curve to a specific point $t_{0} \in \mathbb{R}$ and realising  that $\InHom(\{pt\} ,M) =M$.  \\

Moreover, as we are dealing with morphisms between supermanifolds one can describe everything locally in terms of coordinates. With that in mind, let us consider some coordinate system $(x^{A}) = (x^{\mu}, \theta^{i})$ on $U \subset M$. We also employ a  (global) coordinate system $(t)$ on $\mathbb{R}$. We will not write out explicitly any coordinate system on $S$, or in other words we just localise the S-curves on $M$. Then we can describe any S-curve as

\begin{equation*}
 (x^{A}\circ \gamma_{S})(t) = (x_{S}^{\mu}(t) , \theta_{S}^{i}(t)),
\end{equation*}

which is a system of even and odd functions in $C^{\infty}(\mathbb{R}) \times C^{\infty}(S)$.  The statement that a S-curve $\gamma_{S}$ passes through a point $m \in M(S)$ means $\gamma_{S}(0) = m$ which locally on $M$ is equivalent to the specification of the system of even and odd functions  $(x_{S}^{\mu}(0) , \theta_{S}^{i}(0))$.

\begin{remark}
Supercurves are also found in the literature as morphisms  $\Hom(\mathbb{R}^{1|1},M)$.  These morphisms should not be confused with the notion of an S-curve as used in this work. Importantly it is not sufficient to consider just supercurves in the proceeding constructions.  We will further comment on this in a later section.
\end{remark}

\begin{definition}
Two S-curves $\gamma_{S}$ and $\delta_{S} \in \InHom(\mathbb{R}, M)(S)$ are said to be at \emph{contact to order k at} $m \in M(S)$ if and only if $\gamma_{S}(0) = \delta_{S}(0) = m$ and for every function $f \in C^{\infty}(M)$  we have

\begin{equation*}
J^{k}_{0}(f \circ \gamma_{S}) = J^{k}_{0}(f \circ \delta_{S}).
\end{equation*}

In this case we will write $j_{m}^{k}\gamma_{S} = j_{m}^{k}\delta_{S}$.  An  equivalence class of this relation is a  \emph{(k,S)-jet} from $\mathbb{R}$ to $M$ and we will  denote these by $[\gamma]_{S}^{k}$.
\end{definition}

\begin{remark}
The compositions $f \circ \gamma_{S}$ and $f \circ \delta_{S}$ make sense and are elements in $\InHom(\mathbb{R}, \mathbb{R}^{1|1} )(S)$, that is superfunctions on $\mathbb{R}$ (parameterised by $S$) in the language introduced in the previous section. Thus for a specified $S$ the composition describes the evaluation of the function $f$ over a family of S-points parameterised by $t \in \mathbb{R}$.
\end{remark}

\noindent \textbf{Notation} The set of all (k,S)-jets from $\mathbb{R}$ to $M$ with target $m \in M(S)$ will be denoted by $J_{0}^{k}(\mathbb{R}, M(S))_{m}$. The set of all (k,S)-jets from $\mathbb{R}$ to $M$  will be denoted as  $J_{0}^{k}(\mathbb{R}, M(S))$. We will denote the map between sets that assigns to an S-curve a  (k,S)-jet as
\begin{eqnarray}
\InHom(\mathbb{R}, M)(S) &\stackrel{j_{S}^{k}}{\longrightarrow} & J_{0}^{k}(\mathbb{R}, M(S)) \\
\nonumber \gamma_{S} &\mapsto & [\gamma]_{S}^{k}.
\end{eqnarray}
 Directly from Lemma \ref{lem:naturality} we are led to the following construction.

\begin{construction}
 Given an arbitrary $\psi \in \Hom(P,S)$ we have the induced map of sets
\begin{eqnarray}
\nonumber \Psi^{(k)}: J^{k}_{0}(\mathbb{R}, M(S)) &\rightarrow& J^{k}_{0}(\mathbb{R}, M(P))\\
\nonumber [\gamma]_{S}^{k} &\mapsto& [\gamma]_{P}^{k} := [\gamma \circ(\psi \times \Id_{\mathbb{R}})]_{S}^{k} =[\gamma]_{S}^{k}\circ \psi,
\end{eqnarray}
\end{construction}

 All these proceeding considerations leads the following definition.
 \begin{definition}
 \emph{The k-th order tangent bundle} $T^{(k)}M$ of a supermanifold $M$ is the generalised supermanifold defined by $(T^{(k)}M)(\bullet) := J_{0}^{k}(\mathbb{R}, M(\bullet)) \in \widehat{\catname{SM}}$.
 \end{definition}

 \noindent \textbf{Warning} The k-th order tangent bundle must be viewed  as a  \emph{generalised supermanifold} at this stage.  We need to address the \emph{naturality}  of this construction and then the \emph{representability}. Furthermore, the k-th order tangent bundle should not be confused with the k-th order iterated tangent bundle $T^{k}M := T T \cdots TM$  (k-times).\\

Given an arbitrary $\psi \in \Hom(P,S)$ we have the following diagram which must be commutative if the construction of the k-th order tangent bundle is to be natural.

\begin{diagram}[htriangleheight=30pt ]
\InHom(\mathbb{R},M)(S)& \rTo^{\Psi} & \InHom(\mathbb{R},M)(P)\\
\dTo^{j_{S}^{k}} & &  \dTo^{j_{P}^{k}}  \\
(T^{(k)}M)(S) &\rTo^{\Psi^{(k)}} & (T^{(k)}M)(P)
\end{diagram}

 In other words, one wants $j^{k} : \InHom(\mathbb{R},M) \rightarrow T^{(k)}M$ to be a natural transformation between the respective generalised supermanifolds. The map $\Psi$ is given by $\gamma_{S}\mapsto \gamma_{P} := \gamma_{S} \circ(\psi \times \Id_{\mathbb{R}})$. The map $\Psi^{(k)}$ is given by $[\gamma]_{S}^{k} \mapsto [\gamma]_{P}^{k} := [\gamma]_{S}^{k} \circ \psi$. Thus we need to show that $j_{P}^{k} \circ \Psi =\Psi^{(k)} \circ j_{S}^{k}$ in order to prove we have a natural transformation.

\begin{proposition}\label{prop:naturality}
The construction of the  k-th order tangent bundle $T^{(k)}M$ is natural.
\end{proposition}

\begin{proof}
 Let $f \in C^{\infty}(M)$ be an arbitrary function on $M$ and consider $f\circ \gamma_{S}  \in \InHom(\mathbb{R},\mathbb{R}^{1|1})(S)$. Then from Construction \ref{con:morphism} we have $f\circ \gamma_{S} \mapsto f \circ \gamma_{P} := f \circ \left( \gamma_{S} \circ (\psi \times \Id_{\mathbb{R}}) \right)$, where $\psi \in \Hom(P,S)$. Then one can apply Lemma \ref{lem:naturality} to obtain that $j_{m \circ \psi}^{k} \gamma_{P} = (j_{m}^{k} \gamma_{S}) \circ \psi$. Then passing to the equivalence classes establishes the proposition.
\end{proof}

\begin{lemma}\label{lem:local coords}
Let $(x^{A}) = (x^{\mu}, \theta^{i})$ be a coordinate system on $ U \subset M$. Then two S-curves $\gamma_{S}, \delta_{S} \in \InHom(\mathbb{R}, M)(S)$ are at contact to order $r$ at $m \in M(S)$ if and only if
\begin{equation}\label{eqn:partial derivatives}
\frac{\partial^{k}}{\partial t^{k}} (x^{A} \circ \gamma_{S})(0) = \frac{\partial^{k}}{\partial t^{k}} (x^{A} \circ \delta_{S})(0),
\end{equation}
where $k=0,1,2, \cdots ,r$ and for all coordinate functions $x^{A}$.
\end{lemma}

\begin{proof}
If $j_{m}^{r}\gamma_{S} = j_{m}^{r}\delta_{S}$ then it is clear that \ref{eqn:partial derivatives} holds as each coordinate $x^{A}$ is a function on $M$.  In the other direction, assume that \ref{eqn:partial derivatives} holds. Let $f \in C^{\infty}(M)$ be an arbitrary function. Any  function $f$ on $M$ has a coordinate expression $f(x)$. Then using Fa\`{a} di Bruno's formula \cite{FaadiBruno:1855} (repeated application of the chain rule) it is easy to see that all derivatives up to order r of $f \circ \delta_{S}$ at $t=0$ depend only on the partial derivatives of $f$ at $\gamma_{S}(0)$ up to order r and on \ref{eqn:partial derivatives}. Thus $j_{m}^{r}\gamma_{S} = j_{m}^{r}\delta_{S}$.
\end{proof}

\begin{construction}\label{con:change of base}
Consider $\varphi  \in \Hom(M,N)$. Then there is an induced map of S-points
\begin{eqnarray}
\nonumber \Phi: M(S) &\rightarrow& N(S)\\
\nonumber m &\mapsto& \varphi \circ m,
\end{eqnarray}
which in turn induces a map
\begin{eqnarray}
\nonumber T^{(k)}\Phi: (T^{(k)}M)(S) &\rightarrow& (T^{(k)}N)(S)\\
\nonumber [\gamma]_{S}^{k} &\mapsto& [\varphi \circ \gamma]_{S}^{k}.
\end{eqnarray}
\end{construction}

\begin{theorem}\label{thrm:k-th order tangent bundle}
The k-th order tangent bundle $T^{(k)}M$ of a supermanifold $M$ is representable. Moreover as supermanifolds we have a series of affine fibrations  $T^{(k)}M \rightarrow T^{(k-1)}M \rightarrow \cdots \rightarrow TM \rightarrow M$, where we have $T^{(1)}M = TM$ and $T^{(0)}M = M$.
\end{theorem}

\begin{proof}
Lemma \ref{lem:local coords} allows us  to describe the elements of $\left(T^{(k)}M\right)(S)$ locally on $M$ via coordinate systems using the Taylor expansion of the coordinate expression for the generating S-curves about the point $t=0$. Thus, with respect to any coordinate system $(x^{A})$ on $U \subset M$  an element of the set $\left(T^{(k)}U\right)(S)$ is a tuple of the form

\begin{equation}\label{eqn:coords}
 (x^{A}_{S}, \dot{x}^{A}_{S} , \ddot{x}^{A}_{S}, \cdots, \kAbove{k}{x^{A}_{\!S}}),
\end{equation}

which is an array of collections of functions on $S$ defined as

\begin{equation*}
\kAbove{r}{x^{A}_{S}} := \frac{1}{r!} \left.\frac{\partial^{r}}{\partial t^{r}}(x^{A}\circ \gamma_{S})\right|_{t=0}.
\end{equation*}

  Simple counting of the number of coordinates shows that if the supermanifold $M$ is of dimension $(n|m)$ (both are finite) then the array of functions given by \ref{eqn:coords} is of dimension $((k+1)n | (k+1)m)$.   Now, let $\varphi \in \Hom(U,V)$ be a morphism between two superdomains in the neighborhood of some point on $|M|$. Let us equip the superdomain  $V \subset M$ with local coordinates $(y^{a})$. Then as usual we have $(y^{a} \circ \varphi) = \varphi^{*}y^{a} = y^{a}(x)$. Then at the level of coordinates the S-points transform as
\begin{equation*}
\varphi^{*}y_{S}^{a}(x) = (y^{a} \circ \varphi )\circ m =y^{a}(x) \circ m.
\end{equation*}
One can deduce the transformation rules for the coordinate expressions of the S-points on the k-th order tangent bundle using Construction \ref{con:change of base};

\begin{equation*}
y_{S}^{a} \circ \Phi=(y^{a}\circ \varphi \circ \gamma_{S})(0), \hspace{25pt}
  \kAbove{r}{y^{a}_{\!S}} \circ T^{(r)}\Phi = \left.\frac{1}{r!}\frac{\partial^{r}}{\partial t^{r}} (y^{a}\circ \varphi \circ \gamma_{S})(t)\right|_{t=0},
\end{equation*}

where $0< r\leq k$. By appealing to Fa\`{a} di Bruno's formula we note that $ \kAbove{r}{y^{a}_{\!S}} \circ T^{(r)}\Phi $ depends only on $\kAbove{l}{x^{A}_{S}}$ where $l\leq r$ and polynomially. Furthermore, each term is such that $r = \sum l_{i}$. Importantly it is easy to see that transformation rules preserve the Grassmann parity. This establishes that $T^{(k)}U$, for any superdomain $U \subset M$ is representable as each S-point is described by a finite array of functions on $S$ and that changes of coordinates on $U$ induce well-defined transformation rules for the corresponding coordinate expressions of S-points. Thus we conclude that the k-th order tangent bundle of a supermanifold is representable.  Furthermore we  have a series  natural of affine fibrations   $T^{(k)}M\rightarrow T^{(k-1)}M \rightarrow \cdots \rightarrow TM \rightarrow M$ between supermanifolds induced by the transformation rules for the coordinate expression of the S-points.
\end{proof}

\begin{remark}
The series of projections $ \pi_{l}: T^{(l)}M \rightarrow T^{(l-1)}M$ can be directly understood via  the map $ [\gamma]_{S}^{k} \rightarrow [\gamma]_{S}^{k-1}$ on the equivalence classes of the S-curves. \\
\end{remark}

\begin{corollary}\label{corr:tangent functor}
From Construction \ref{con:change of base} and Theorem \ref{thrm:k-th order tangent bundle} we have the  functor $T^{(k)}: \catname{SM} \rightarrow \catname{SM}$.  Moreover, as  any morphism of supermanifolds $S \times \mathbb{R} \rightarrow M_{1} \times M_{2}$ coincides with a pair of morphisms  $S \times \mathbb{R} \rightarrow M_{1}$ and $S \times \mathbb{R} \rightarrow  M_{2}$ the functor $T^{(k)}$ preserves products, just as in the classical case.
\end{corollary}

\begin{example}
To illustrate the above theorem let us concentrate on $T^{(2)}M$. Let us equip $M$ with local coordinates $(x^{A})$. Then the S-points of  $T^{(2)}M$ can locally on $M$ be described by $(x^{A}_{S}, \dot{x}^{A}_{S}, \ddot{x}^{A}_{S})$. Now, let us consider a local diffeomorphism $\varphi : M \rightarrow M$, which in local coordinates is represented by $x^{A'} \circ \varphi = x^{A'}(x) $. Then we can calculate the effects of this change of local coordinates on the S-points
\begin{eqnarray}
\nonumber x^{A'}_{S} &=&  (x^{A'} \circ \varphi) \circ \gamma_{S}(0),\\
\nonumber \dot{x}^{A'}_{S} &=& \left.\frac{\partial (x^{B} \circ \gamma_{S})(t)}{\partial t} \right|_{t=0}\cdot\frac{\partial (x^{A'}\circ \varphi  )}{\partial x^{B}}\circ \gamma_{S}(0)= \dot{x}^{B}_{S} \cdot\left( \frac{\partial x^{A'}(x)}{\partial x^{B}} \circ \gamma_{S}(0)\right),\\
\nonumber \ddot{x}^{A'}_{S} &=&\left. \frac{1}{2}\frac{\partial^{2} (x^{B} \circ \gamma_{S})(t)}{\partial t^{2}} \right|_{t=0} \cdot \frac{\partial (x^{A'} \circ \varphi)}{\partial x^{B}} \circ \gamma_{S}(0) + \frac{1}{2} \left.\frac{\partial (x^{B} \circ \gamma_{S})(t)}{\partial t} \frac{\partial (x^{C} \circ \gamma_{S})(t)}{\partial t}     \right|_{t=0}\cdot \frac{\partial^{2} (x^{A'}\circ \varphi  )}{ \partial x^{C}\partial x^{B}}\circ \gamma_{S}(0)\\
\nonumber &=& \ddot{x}^{B}_{S} \left( \frac{\partial x^{A'}(x)}{\partial x^{B}}  \circ \gamma_{S}\right) + \frac{1}{2} \dot{x}^{B}_{S} \dot{x}^{C}_{S}\cdot \left(\frac{\partial^{2} x^{A'}(x)}{\partial x^{C} \partial x^{B}}  \circ \gamma_{S}(0)  \right).
\end{eqnarray}
We see that we can then deduce that changes of coordinates on $T^{(2)}M$  are (using standard abuses of notation)
\begin{equation*}
x^{A'} = x^{A'}(x),\hspace{15pt}
\dot{x}^{A'} = \dot{x}^{B} \frac{\partial x^{A'}}{\partial x^{B}},\hspace{15pt}
 \ddot{x}^{A'} = \ddot{x}^{B} \frac{\partial x^{A'}}{\partial x^{B}} + \frac{1}{2}\dot{x}^{B} \dot{x}^{C} \frac{\partial^{2} x^{A'}}{\partial x^{C} \partial x^{B}},
\end{equation*}
which are of course of the same form as the classical case.
\end{example}

\begin{remark}
The standard algebraic approach to defining the tangent bundle of a supermanifold is to define it  in terms of the derivations on the algebra of functions on the supermanifold. The derivations of the functions form a sheaf of locally free modules on the supermanifold and so define algebraically a vector bundle structure. In this note we have a kinematic definition of the total space of the tangent bundle in terms of equivalence classes of S-curves.  This construction is already well-know, however the explicit construction of the higher order tangent bundles in this way appears to be missing from the literature.
\end{remark}

\begin{statement}
At the operational level  of local coordinates the  k-th order tangent bundle of a supermanifold can be defined via Taylor expansions of the coordinate expression for the generating curves. Moreover, the k-th order tangent bundle of a supermanifold  can be understood in terms of (adapted) local coordinates and their transformation laws in exactly the same way as the classical case.
\end{statement}

\section{The graded structure of the k-th order tangent bundle}
The \emph{homotheties} on $T^{(k)}M$, that is  particular smooth actions of the multiplicative semigroup $(\mathbb{R},\cdot)$, can also be understood geometrically. Specifically, we have the canonical action of $\mathbb{R}$ on $\mathbb{R}$ as
\begin{eqnarray}
\nonumber \mathrm{g} : \mathbb{R} \times \mathbb{R} &\rightarrow & \mathbb{R}\\
\nonumber (\lambda , t) &\mapsto& \mathrm{g}(\lambda, t) = \lambda t.
\end{eqnarray}

We will write $\mathrm{g}(\lambda, t) = \mathrm{g}_{\lambda}(t)$.

\begin{construction}\label{con:action S-curves}
The above  canonical action of   $(\mathbb{R},\cdot)$ extends to an action on the S-curves viz
\begin{eqnarray}
\nonumber \widehat{\mathrm{h}}_{S} : \mathbb{R} \times \InHom(\mathbb{R}, M)(S) &\rightarrow & \InHom(\mathbb{R}, M)(S) \\
\nonumber (\lambda , \gamma_{S}) &\mapsto&  \gamma_{S} \circ ( \Id_{S} \times \mathrm{g}_{\lambda}).
\end{eqnarray}
\end{construction}

\begin{proposition}\label{prop:naturality action S-curves}
The action of the multiplicative semigroup  $(\mathbb{R}, \cdot)$ on the set of S-curves $\InHom(\mathbb{R}, M)(S)$ is natural.
\end{proposition}
\begin{proof}
Given  $\Psi : \InHom(\mathbb{R}, M)(S) \rightarrow \InHom(\mathbb{R}, M)(P)$ defined as $\gamma_{S} \mapsto \gamma_{S} \circ (\psi \times \Id_{\mathbb{R}})$ for arbitrary $\psi \in \Hom(P,S)$, one can directly show that
\begin{equation*}
\widehat{\mathrm{h}}_{P} \circ (\Id_{\mathbb{R}}\times \Psi ) = \Psi \circ \widehat{\mathrm{h}}_{S}.
 \end{equation*}
 We leave details to the reader as this is a matter of routine.
\end{proof}

\begin{construction}\label{con: action jets}
The action of the multiplicative semigroup  $(\mathbb{R}, \cdot)$ on S-curves extends to a canonical action on  $(T^{(k)}M)(S)$ viz
\begin{eqnarray}
\nonumber \mathrm{h}_{S}: \mathbb{R} \times (T^{(k)}M)(S) &\longrightarrow & (T^{(k)}M)(S)\\
\nonumber (\lambda, [\gamma]_{S}^{k}) &\mapsto & [\widehat{\gamma}(\lambda)]_{S}^{k},
\end{eqnarray}
where $\widehat{\gamma}_{S}(\lambda) :=  \widehat{\mathrm{h}}_{S}(\lambda, \gamma_{S}) = \gamma_{S} \circ (\Id_{S} \times \mathrm{g}_{\lambda})$. It is easy to verify that $\mathrm{h}_{S}(\lambda \mu) = \mathrm{h}_{S}(\lambda)\circ \mathrm{h}_{S}(\mu)$, where we define $\mathrm{h}_{S}(\nu) : (T^{(k)}M)(S) \rightarrow (T^{(k)}M)(S)$ by restriction of the action $\mathrm{h}_{S}: \{\nu\} \times (T^{(k)}M)(S) \longrightarrow  (T^{(k)}M)(S)$.
\end{construction}

\begin{theorem}\label{thrm:natural homotheties}
The  action of the multiplicative semigroup  $(\mathbb{R}, \cdot)$ on $T^{(k)}M$  is natural.
\end{theorem}

\begin{proof}
The requirement to be natural is that the following diagram is commutative:

\begin{diagram}[htriangleheight=30pt ]
\mathbb{R} \times (T^{(k)}M)(S)& \rTo^{(\Id_{\mathbb{R}} \times \Psi^{(k)})} & \mathbb{R} \times (T^{(k)}M)(P)\\
\dTo^{\mathrm{h}_{S}} & &  \dTo^{\mathrm{h}_{P}}  \\
(T^{(k)}M)(S) &\rTo^{\Psi^{(k)}} & (T^{(k)}M)(P)
\end{diagram}
That is we require $\mathrm{h}_{P} \circ(\Id_{\mathbb{R}} \times \Psi^{(k)}) =  \Psi^{(k)} \circ \mathrm{h}_{S}$. Starting with the RHS we obtain $\Psi^{(k)}\left(\mathrm{h}_{S} (\lambda, [\gamma]_{S}^{k})  \right) = \Psi^{(k)}\left([\widehat{\gamma}(\lambda)]_{S}^{k})  \right) = [\widehat{\gamma}(\lambda)]_{S}^{k} \circ \psi$. On the LHS we have $\mathrm{h}_{P}\left((\Id_{\mathbb{R}} \times \Psi^{(k)})(\lambda, [\gamma]_{S}^{k}) \right) = \mathrm{h}_{P}(\lambda , [\gamma]_{P}^{k}) = [\widehat{\gamma}(\lambda)]_{P}^{k} = [\widehat{\gamma}(\lambda)]_{S}^{k} \circ \psi$, which follows from Proposition \ref{prop:naturality}.
\end{proof}

The above theorem states that we have a natural transformation $\mathrm{h}(\lambda):  T^{(k)}M \rightarrow T^{(k)}M $ when we fix a point $\lambda \in \mathbb{R}$ and  consider the k-th order tangent bundle of a supermanifold as a functor via the Yoneda embedding. Thus via Yoneda's lemma, we actually have a well-defined  action of $(\mathbb{R}, \cdot)$ on the supermanifold $T^{(k)}M$.\\

This action is best understood via local coordinates. It is only a  matter of applying the chain rule carefully to show that on the level of local coordinates the action on the S-points is given by
\begin{eqnarray}
\nonumber (\mathrm{h}_{S}^{ \lambda})^{*} x_{S}^{A} &=& x_{S}^{A},\\
\nonumber (\mathrm{h}_{S}^{\lambda})^{*}\kAbove{r}{x_{S}^{A}} &=& (\lambda)^{r} \kAbove{r}{x_{S}^{A}}.
\end{eqnarray}
One can then pass to the supermanifold $T^{(k)}M$ itself  using Yoneda's lemma. Typographically one only has to drop the subscript $S$. Note that $\mathrm{h}(0)$ corresponds to  the projection $\pi : T^{(k)}M \rightarrow M$ as genuine supermanifolds.\\

 The corresponding weight vector field  $\Delta \in \Vect(T^{(k)}M)$  is given by  the value of its action on any function $F\in C^{\infty}(T^{(k)}M)$ at an S-point;

 \begin{equation}
 (\Delta F)\circ [\gamma]_{S}^{k} := \left.\frac{d}{dt}\right|_{\lambda =1} F\circ [\widehat{\gamma}(\lambda)]_{S}^{k} =    \left(\sum_{r=1}^{k} r \kAbove{r}{x^{A}} \frac{\partial F }{\partial \kAbove{r}{x^{A}}}\right )\circ [\gamma]_{S}^{k}.
 \end{equation}

\begin{statement}
At the level of local coordinates the action of the multiplicative semigroup $(\mathbb{R}, \cdot)$ on the k-th order tangent bundle of a supermanifold is identical to the classical case. Moreover, we have the structure of a graded  super bundle in the language of \cite{Grabowski:2011}.
\end{statement}

\section{Comparison with other notions of curves}
 \noindent \textbf{Supercurves}\\
 \emph{Supercurves} are  understood as morphisms  in $\Hom(\mathbb{R}^{1|1} ,M) = \InHom(\mathbb{R}, M)(\mathbb{R}^{0|1})$ appear in the literature as the simplest generalisation of classical curves that can ``feel" the odd dimensions of a supermanifold, see for example \cite{Garnier:2012,Goertsches:2008} where they have been put to good use on Riemannian supermanifolds. However, supercurves are not sufficient to recover the notion of the k-th order tangent bundle, unless $M$ has at most one odd dimension. In fact, one cannot define  the tangent bundle in this way.\\

 Explicitly in any local coordinate system on $M$ supercurves $c \in \Hom(\mathbb{R}^{1|1} ,M) $ are of the form  $ c^{*}(x^{\mu}, \theta^{i}) = (x^{\mu}(t) , \tau \: u^{i}(t))$, where we have picked global coordinates $(t, \tau)$ on $\mathbb{R}^{1|1}$. Then Taylor expanding about $t =0$ to order one gives the collection of  functions

 \begin{equation}(x^{\mu}_{\mathbb{R}^{0|1}} , \theta^{i}_{\mathbb{R}^{0|1}} := \tau \: u^{i}_{\mathbb{R}^{0|1}} ,\dot{x}^{\mu}_{\mathbb{R}^{0|1}},  \dot{\theta}^{i}_{\mathbb{R}^{0|1}} :=\tau \: \dot{u}^{i}_{\mathbb{R}^{0|1}} ),
 \end{equation}

 on $\mathbb{R}^{0|1}$ which we interpret as the $\mathbb{R}^{0|1}$-points of the tangent bundle $TM$. Now consider a change of coordinates on $M$. As changes of coordinates respect the Grassmann parity we have the following expressions

 \begin{eqnarray}
 x^{\mu'}(x,\theta) &=& x^{\mu'}(x) + \sum_{Even} \frac{1}{l!}\theta^{i_{1}} \cdots \theta^{i_{l}}x^{\mu'}_{i_{l}\cdots i_{1}}(x)\\
 \nonumber \theta^{j'}(x,\theta) &=& \sum_{Odd} \frac{1}{l!}\theta^{i_{1}} \cdots \theta^{i_{l}}w^{j'}_{i_{l}\cdots i_{1}}(x).
 \end{eqnarray}

 It is not hard to see that the induced transformation laws of the $\mathbb{R}^{0|1}$-points of $TM$ are\\

 \begin{tabular}{ll}
 $x^{\mu'}_{\mathbb{R}^{0|1}} = (x^{\mu'}(x))\circ c(o,\tau)$, & $ \theta^{j'}_{\mathbb{R}^{0|1}} = (\theta^{j}w_{j}^{i'}(x)) \circ c(o,\tau)$,\\
 $\dot{x}^{\mu'}_{\mathbb{R}^{0|1}} = \dot{x}^{\nu}_{\mathbb{R}^{0|1}} \cdot \left(\frac{\partial x^{\mu'}(x)}{\partial x^{\nu}} \right) \circ c(o,\tau),$ & $ \dot{\theta}^{j'}_{\mathbb{R}^{0|1}} = \dot{\theta}^{i}_{\mathbb{R}^{0|1}}\cdot (w_{j}^{i'}(x))\circ c(0,\tau)$,
 \end{tabular}\\

taking into account the fact that $\tau^{2} = 0$. Thus we see that $\mathbb{R}^{0|1}$-points only really capture the vector bundle structure of $M$ and really misses the full structure of  $TM$ as a natural bundle.  This is of course better than the case of considering classical curves only, but is clearly not sufficient.  \\

\begin{remark}
Taylor expanding the local expressions of supercurves about $t = \tau =0$ is even worse than the situation above as this totally misses the odd directions on $M$ and hence we only construct the higher tangent bundles of the reduced manifold underlying $M$.
\end{remark}

\noindent \textbf{Higher dimensional supercurves}\\
Via the work of Schwarz  and Voronov \cite{Shvarts:1984,Voronov:1984} we know that it is sufficient to consider $\Lambda$-points, that is one can probe $\InHom(\mathbb{R},M)$ with supermanifolds of the form $\mathbb{R}^{0|l}$ for $l \geq 1$. One could then try to think  about ``curves" as being  in $\Hom(\mathbb{R}^{1|l} , M)$ for \emph{some} $l$. There are essentially two generic options for specifying $l$;
\renewcommand{\theenumi}{\roman{enumi})}%
\begin{enumerate}
\item One could think of $l$ being ``large enough"  (or even infinite!) so that all the computations are consistent. This is related to the DeWitt--Rogers  approach to supermanifolds \cite{DeWitt:1984,Rogers:1980}, where a supermanifold is a manifold modeled on  a Grassmann algebra equipped with some suitable (non-Hausdorff) topology. One has to further make restrictions on the classes of functions one considers in order to properly construct  supermanifolds in the DeWitt--Rogers  approach\footnote{The interested reader should consult Rogers \cite{Rogers:2007} for a clear comparison of the various approaches to supermanifolds. A brief account can be found in \cite{Helein:2009}.}.
\item One can keep $l$ arbitrary by  considering families of morphisms and rephrasing the constructions in the language of category theory, as we have done here.
\end{enumerate}

The problem with the first approach is the freedom in choosing an appropriate $l$, though of course it maybe a useful thing to do for calculational purposes. Once one has found a minimal $l$ suitable for the problem at hand, one can  always find a larger number that will also be suitable. Moreover, no construction using curves should depend in any critical way on this number provided it is large enough. The dependance of the number of odd dimensions to a curve should be functorial and so we are lead back to the philosophy of second approach listed above.\\

\noindent \textbf{Superpaths}\\
There is also the notion of a \emph{superpath}  as  maps living in $\InHom(\mathbb{R}^{1|1}, M)$. Superpaths were used to relate parallel transport and Quillen superconnections by Dumitrescu \cite{Dumitrescu:2008}. Note that we have $\Hom(S \times \mathbb{R}^{1|1}, M)  \simeq \Hom (S \times \mathbb{R}, \Pi TM) $. Here $\Pi TM = \InHom(\mathbb{R}^{0|1},M)$ is the antitangent bundle, and can be constructed from (the total space of) the tangent bundle by shifting the parity of the fibre coordinates. In short (parameterised) superpaths  on a supermanifold are S-curves on its antitangent bundle.  \\

\begin{remark}
One can of course iterate the above identification to show $\InHom(\mathbb{R}^{1|p},M)(S) \simeq \InHom(\mathbb{R}^{1|p-1},\Pi T M)(S) \simeq \InHom(\mathbb{R},(\Pi T)^{p}M)(S)$. That is (parameterised) higher dimensional superpaths on a supermanifold can be understood as S-curves on its  appropriately iterated antitangent bundle.
\end{remark}

We then define  the following supermanifold using the constructions in this note;
\begin{equation*}
\mathbb{T}^{(k)}M := T^{(k)}(\Pi TM).
\end{equation*}

Then one can check via local coordinates that we have the diffeomorphism $T^{(k)}(\Pi TM) \simeq \Pi T(T^{(k)}M)$. Thus we see that functions on the supermanifold of k-th jets of superpaths (with source zero) are (pseudo)differential forms on the k-th order tangent bundle. There is also a double homogeneity structure \cite{Grabowski:2011} here described by the two  commuting weight vector fields \\

\begin{tabular}{p{10cm} p{6cm}}
 \begin{eqnarray}
\nonumber \Delta_{1} &=& dx^{A}\frac{\partial}{\partial dx^{A}} + \sum_{r=1}^{k}  d\kAbove{r}{x^{A}}\frac{\partial}{\partial d  \kAbove{r}{x^{A}}}\\
\nonumber \Delta_{2} &=& \sum_{r=1}^{k} r \left( \kAbove{r}{x^{A}} \frac{\partial}{\partial \kAbove{r}{x^{A}}}  + d\kAbove{r}{x^{A}} \frac{\partial}{\partial d\kAbove{r}{x^{A}}}\right)
\end{eqnarray}
&

\begin{diagram}[htriangleheight=15pt ]
  &\mathbb{T}^{(k)}M& \\
\ldTo(1,2) & &  \rdTo(1,2)  \\
 T^{(k)}M &  & \Pi TM\\
&\rdTo(1,2)  \ldTo(1,2)&\\
&M&
\end{diagram}
\end{tabular}
where we have picked natural local coordinates $(x^{A}, dx^{A} , \kAbove{r}{x^{A} }, d\kAbove{r}{x^{A}})$ for $1\leq r \leq k$. Together with the de Rham differential
\begin{equation*}
d = dx^{A}\frac{\partial}{\partial x^{A}} + \sum_{r=1}^{k}  d\kAbove{r}{x^{A}}\frac{\partial}{\partial   \kAbove{r}{x^{A}}},
\end{equation*}
we have the following (super) Lie algebra
\begin{equation}
[d,d] =0, \hspace{10pt} [\Delta_{1}, \Delta_{2}] =0, \hspace{10pt} [\Delta_{1}, d] =d, \hspace{10pt} [\Delta_{2}, d] =0.
\end{equation}
We see that this lie algebra is given by the central extension of the Lie algebra of $\InDiff(\mathbb{R}^{0|1})$, whose infinitesimal action on $\Pi T(T^{(k)}M)$ defines $\Delta_{1}$ and $d$, by the abelian Lie algebra generated by $\Delta_{2}$ that originates from the action of $(\mathbb{R}, \cdot)$ on $T^{(k)}(\Pi TM)$.\\

One can also understand the canonical (integrable) higher almost tangent structure here in a similar way. Specifically the transformation $t' \mapsto t$ and $\tau' \mapsto \tau+ \epsilon \:t$, where $\epsilon$ in an odd parameter and $(t ,\tau)$ are global coordinates on $\mathbb{R}^{1|1}$ gives rise to the homological vector field
\begin{equation*}
J = \sum_{r=0}^{k-1} d\kAbove{r}{x^{A}}\frac{\partial}{\partial   \kAbove{r+1}{x^{\:\:\:\:A}}},
\end{equation*}

which we recognise to be the required higher almost tangent structure. It is then easy to verify that
\begin{equation}
[J,J] = 0, \hspace{15pt} [\Delta_{1},J] = J, \hspace{15pt} [\Delta_{2},J] = -J, \hspace{15pt} [d, J]=0.
\end{equation}

By passing to total weight, which is essential described by the sum of the two weight vector fields $\Delta = \Delta_{1}+ \Delta_{2}$ we arrive at the following Lie algebra
\begin{equation}
[d,d]=0, \hspace{15pt} [\Delta, d] = d, \hspace{15pt} [\Delta, J]=0,\hspace{15pt} [d,J]=0,\hspace{15pt} [J,J]=0,
\end{equation}

and thus we see that this Lie algebra is again a central extension of the Lie algebra of $\InDiff(\mathbb{R}^{0|1})$ but this time by the abelian Lie algebra generated by a single odd element. Note that in general the vanishing of the self-commutator of an odd vector field is a non-trivial condition.

\begin{remark}
To the author's knowledge, interpreting the canonical higher almost tangent structure on a higher order tangent bundle in this way is new and the consequences await to be properly explored.
\end{remark}

\begin{statement}Putting the above observations together, if we restrict attention to finite dimensional supermanifolds and their morphisms, then there is no single privileged  supermanifold that plays the role of the source  of curves in a completely  satisfactory way: one seems rather forced to employ the internal Homs.
\end{statement}

\section*{Acknowledgments}
The author would like to thank J. Grabowski for suggesting that their should be a natural geometric definition of the k-th order tangent bundle of a supermanifold. A special thank you goes to M. Rotkiewicz and M.  J\'{o}\'{z}wikowski. The author would also like to thank the organisers of  the seminar ``Group actions and manifolds" 8-11 November 2013, Stefan Banach International Mathematical Center, Warsaw, Poland where part of this work was announced. The author graciously acknowledges the support of the Warsaw Center of Mathematics and Computer Science.

\null
\vfill
\begin{center}
Andrew James Bruce\\
\small \emph{Institute of Mathematics, Polish Academy of Sciences,}\\ \small \emph{ul. \'{S}niadeckich 8, P.O. Box 21}\\ \small \emph{00-956 Warszawa, Poland},\\ \small email:\texttt{andrewjamesbruce@googlemail.com}
\end{center}

\end{document}